\documentclass[letterpaper, 10 pt, conference]{ieeeconf}  

\IEEEoverridecommandlockouts                              

\usepackage{amsmath} 
\usepackage{amssymb}
\usepackage{mathtools}

\usepackage{xcolor}
\usepackage{pgfplots}
\usepackage{tikz}
\usetikzlibrary{decorations.pathmorphing,arrows.meta,calc,angles,quotes,positioning}

\usepackage[isbn=false, doi=false, maxnames=99, backend=biber, date=year]{biblatex}
\newcommand{\vect}[1]{\boldsymbol{#1}}
\newcommand{\mat}[1]{\boldsymbol{#1}}

\newcommand{\vx}{\vect{x}}
\newcommand{\vu}{\vect{u}}
\newcommand{\vy}{\vect{y}}
\newcommand{\vv}{\vect{v}}

\newcommand{\ve}{\vect{e}}
\newcommand{\veta}{\vect{\eta}}

\newcommand{\bx}{\bar{x}}
\newcommand{\bu}{\bar{u}}
\newcommand{\bvx}{\bar{\vx}}
\newcommand{\bvu}{\bar{\vu}}

\newcommand{\ba}{\bar{a}}
\newcommand{\ta}{\tilde{a}}

\newcommand{\mOmega}{\mat{\Omega}}
\newcommand{\mA}{\mat{A}}
\newcommand{\mV}{\mat{V}}
\newcommand{\mLambda}{\mat{\Lambda}}
\newcommand{\tmLambda}{\tilde \mLambda}

\newcommand{\mB}{\mat{B}}
\newcommand{\mH}{\mat{H}}

\newcommand{\mK}{\mat{K}}

\newcommand{\mI}{\mat{I}}

\newcommand{\mM}{\mat{M}}

\newcommand{\mcalT}{\mat{\mathcal{T}}}

\newcommand{\Char}{\phi}
\newcommand{\invChar}{\psi}

\newcommand{\Reals}{\ensuremath{\mathbb{R}}}

\newcommand{\dens}{\varrho}

\DeclareMathOperator{\diagC}{\mathrm{diag}}

\newcommand{\dm}{\mathrm{d}}
\newcommand{\tint}{\textstyle \int}

\definecolor{myBlue}{RGB}{0, 0, 153} 		
\definecolor{myRed}{RGB}{255, 0, 0} 
\definecolor{myGreen}{RGB}{0, 153, 0} 		

\newtheorem{remark}{Remark}
\newtheorem{theorem}{Theorem}
\newtheorem{Lemma}{Lemma}

\usepackage{environ}
\NewEnviron{subalign}[1][]{%
	\begin{subequations}
		\if\relax\detokenize{#1}\relax\else\label{#1}\fi
		\begin{align}
			\BODY
		\end{align}
	\end{subequations}
}

\title{\LARGE \bf
Flatness-based control of a Timoshenko beam
}

\author{\underline{Simon Schmidt}$^1$, Nicole Gehring$^2$, Abdurrahman Irscheid$^3$
\thanks{
	This research was funded in part by the Austrian Science Fund (FWF) [I~6519-N] as well as the Deutsche Forschungsgemeinschaft (DFG, German Research Foundation) – project no. 517291864.
	For open access purposes, the author has applied a CC BY public copyright license to any author accepted manuscript version arising from this submission.}
\thanks{$^{1}$Simon Schmidt is with the Institute of Control Systems,
		Johannes Kepler University Linz, Linz, Austria
        {\tt\small simon.schmidt@jku.at}}%
\thanks{$^{2}$Nicole Gehring is with the Chair of Systems Theory and Control Engineering,
		Otto von Guericke University Magdeburg, Magdeburg, Germany
        {\tt\small nicole.gehring@ovgu.de}}%
\thanks{$^{3}$Abdurrahman Irscheid is with the Chair of Systems Theory and Control Engineering,
        Saarland University, Saarbrücken, Germany
        {\tt\small a.irscheid@lsr.uni-saarland.de}}%
}

\begin{document}

\maketitle
\thispagestyle{empty}
\pagestyle{empty}

\begin{abstract}

The paper presents an approach to flatness-based control design for hyperbolic multi-input systems, building upon the hyperbolic controller form (HCF).
The transformation into HCF yields a simplified system representation that considerably facilitates the design of state feedback controllers for trajectory tracking.
The proposed concept is demonstrated for a Timoshenko beam and validated through numerical simulations, demonstrating trajectory tracking and closed-loop stability.

\end{abstract}


\section{Introduction}

The flatness-based control of nonlinear ordinary differential equations (ODEs) is based on the parametrization of all system variables by a flat output \cite{FLIESS1995}.
The parametrization of the system's input allows for the derivation of a system description in controller canonical form or Brunovsky normal form, state representations that significantly simplify the design of feedforward and state feedback (tracking) controllers.
This approach also extends to infinite-dimensional systems.
For hyperbolic single-input (SI) systems with dynamic boundary conditions, a flatness-based control concept is introduced in \cite{Woittennek2012a}.
Its central element is the transformation of the system into hyperbolic controller form (HCF), a special state representation in which the system dynamics are represented by a single boundary-actuated transport equation \cite{Russell1978}, analogous to the chain of integrators in the finite-dimensional case.
In HCF, the control design becomes straightforward, and the resulting control law can be readily expressed as feedback of the original system state through the corresponding, invertible state transformation.

Generalizing the flatness-based control design from SI to multi-input (MI) systems poses additional challenges, similar to the ODE case.
In the latter, the controller form is no longer unique \cite{Isidori1985} as it may be represented by integrator chains of different length.
Moreover, it may involve time derivatives of the input, rendering it a generalized state representation \cite{Rudolph2021}.
In this case, a purely static state feedback is generally insufficient for trajectory tracking.
Similar considerations arise for hyperbolic MI systems, as discussed in \cite{Schmidt2025} for boundary-actuated systems without in-domain coupling.

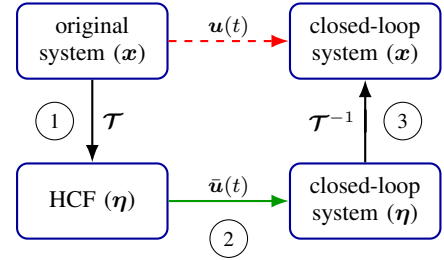
\begin{figure}[t]
	\centering
	\begin{tikzpicture}[
	>=Latex,
	node distance=11mm and 16mm,
	block/.style={draw=myBlue, rounded corners, minimum width=20mm, minimum height=10mm, align=center, thick},
	circ/.style={circle, draw=black, inner sep=1pt, minimum size=5mm, font=\footnotesize, fill=white, text=black, thin},
	lab/.style={font=\footnotesize, inner sep=1pt, text=black}
	]
	
	\small
	
	\node[block] (orig) {original \\ system ($\vx$)};
	\node[block, right=of orig] (CLtop) {closed-loop \\ system ($\vx$)};
	\node[block, below=of orig] (HCF) {HCF ($\veta$)};
	\node[block, right=of HCF] (CLbot) {closed-loop \\ system ($\veta$)};
	
	\draw[dashed, ->, thick, myRed] (orig) -- node[lab, above] {$\vu(t)$} (CLtop);
	
	\draw[->, thick] (orig) -- node[lab, right] {$\; \mcalT$}
	node[circ, xshift=-5mm] {1} (HCF);
	
	\draw[->, thick, myGreen] (HCF) -- node[lab, above] {$\bvu(t)$} 
	node[circ, yshift=-5mm] {2} (CLbot);
	
	\draw[->, thick] (CLbot.north) -- ++(0,7mm) -- node[lab, left] {$\mcalT^{-1} \;$}
	node[circ, xshift=5mm] {3} ([yshift=-7mm]CLtop.south) -- (CLtop.south);
\end{tikzpicture}%
	\caption{Schematic of the HCF-based control design.}
	\label{fig_BSB}
	\vspace*{-10pt}
\end{figure}

This paper extends the flatness-based control design from \cite{Schmidt2025} to hyperbolic MI systems with additional in-domain coupling.
The proposed concept relies on the transformation into HCF, a system representation that is advantageous for feedback design (see Figure~\ref{fig_BSB}).  
This is somewhat similar to the well-known \textit{backstepping} method \cite{Krstic2008}, which also employs a state transformation to simplify the control design.
The main challenge of the HCF-based method lies in identifying a suitable candidate for the HCF state.
As the flat output of systems with multiple inputs is not unique, the same applies to the HCF state \cite{Schmidt2025}.  
Moreover, the HCF is no longer a state representation in the classical sense, in general, as it involves input predictions.
Therefore, a non-static feedback of the HCF state is derived for trajectory tracking.
It is shown that the transformation into HCF is invertible, allowing to express the control law in terms of the original system state.

The flatness-based control concept is illustrated for a Timoshenko beam \cite{Timoshenko1928}, which is governed by two coupled wave equations.
In the studied setup, the beam is clamped at one end and actuated at the opposite end by a boundary force and torque serving as the control inputs (see Figure~\ref{fig_beam}).
This is a prototypical hyperbolic MI system, for which numerous control strategies exist in the literature, including direct boundary assignment \cite{Kim1987}, flatness-based feedforward tracking control \cite{Becker2006}, backstepping approaches \cite{Krstic2006, Redaud2022, Chen2024}, and control in a distributed port-Hamilton framework \cite{Macchelli2004}.

The paper is structured as follows.
Section~\ref{sec_problem} introduces the Timoshenko beam model and specifies the control objective.
In Section~\ref{sec_derivFlatPar} a flatness-based parametrization is derived, from which the HCF is explicitly constructed in Section~\ref{sec_HCF}.
A tracking controller is derived in Section~\ref{sec_ctrl}, ensuring the control objectives are met.
Finally, the results are verified through simulation in Section~\ref{sec_sim}.

\emph{Notation}:
For partitioned vectors $\vv \in \Reals^{n_- + n_+}$ and matrices $\mM \in \Reals^{(n_- + n_+) \times (n_- + n_+)}$, we write
\begin{equation}
	\vv = \begin{bmatrix} \vv^- \\ \vv^+ \end{bmatrix}
	\quad \text{and} \quad 
	\mM = \begin{bmatrix} \mM^{--} & \mM^{-+} \\ \mM^{+-} & \mM^{++} \end{bmatrix}.
\end{equation}
We use a shorthand notation $(\cdot)^\mp$ to compactly refer to expressions that apply to both $(\cdot)^-$ and $(\cdot)^+$.


\section{Problem statement}
\label{sec_problem}

We consider a Timoshenko beam without internal damping, which is described by two coupled wave equations, i.e., second-order partial differential equations (PDEs)
\begin{subalign}[eq_waveeqs]
	\dens(z) S(z) \, \partial_t^2 w(z,t) &= \partial_z \big( \kappa(z) \big( \partial_z w (z,t) - \varphi(z,t) \big) \big) \\
	J(z) \, \partial_t^2 \varphi (z,t) &= \partial_z \big( E(z)I(z) \, \partial_z \varphi (z,t) \big) \nonumber \\
	& \phantom{=} + \kappa(z) \big( \partial_z w(z,t) - \varphi(z,t) \big),
\end{subalign}
where $w(z,t)$ denotes the transverse displacement along the beam axis and $\varphi(z,t)$ is the rotation angle of the normal of the cross-section relative to the beam axis (see Figure~\ref{fig_beam}).
Without loss of generality, we assume $w(z,t)$ and $\varphi(z,t)$ defined on $[0,1] \times [0,\infty)$, which can always be achieved by normalization of the spatial domain.
In \eqref{eq_waveeqs}, $\dens(z)$ denotes the (linear) density, $\kappa(z) = k_s(z) G(z) S(z)$ the product of shear correction factor $k_s(z)$, shear modulus $G(z)$ and cross-section area $S(z)$, $J(z)=\dens(z) I(z)$ the rotary moment of inertia, $I(z)$ the moment of inertia, and $E(z)$ Young's modulus of elasticity.
The beam is assumed to be clamped at $z=0$, i.e., $w(0,t) = 0$ and $\varphi(0,t)=0$.
This implies boundary conditions (BCs)
\begin{equation}
	\label{eq_BCwave_0}
	\partial_t w(0,t) = 0
	\quad \text{and} \quad
	\partial_t \varphi(0,t) = 0.
\end{equation} 
for the corresponding velocities, which are used for the first-order representation in Section~\ref{sec_trafos}.
The end at $z=1$ is fully actuated, leading to the BCs
\begin{subalign}[eq_BCwave_1]
	\kappa(1) \big( \partial_z w(1,t) - \varphi(1,t) \big) &= u_1(t) \\
	E(1)I(1) \, \partial_z \varphi(1,t) &= u_2(t),
\end{subalign}
where the input components $u_1(t)$ and $u_2(t)$ represent a force and a bending moment, respectively.
We consider general initial conditions (ICs) $w(z,0) = \bar w_0(z)$, $\varphi(z,0) = \bar \varphi_0(z)$ and $\partial_t w(z,0) = \bar w_1(z)$, $\partial_t \varphi(z,0) = \bar \varphi_1(z)$.

The objective is to design a state feedback controller that stabilizes the transition of the Timoshenko beam \eqref{eq_BCwave_0}--\eqref{eq_BCwave_1} from an initial stationary configuration $(w(z,t_0) ,\, \varphi(z,t_0)) = (w_0(z) ,\, \varphi_0(z))$ to a desired stationary configuration $(w(z,t_T) ,\, \varphi(z,t_T)) = (w_T(z) ,\, \varphi_T(z))$ in finite time $T = t_T - t_0 \in \Reals^+$.
The control design proceeds in three steps, as illustrated in Figure~\ref{fig_BSB}:
\begin{enumerate}
	\item 
	The HCF for the Timoshenko beam is derived from a flatness-based parametrization of \eqref{eq_waveeqs}--\eqref{eq_BCwave_1}.
	To this end, a sequence of transformations is applied to the hyperbolic MI system to map it into a form beneficial for this purpose.
	
	\item For the system in HCF, the subsequent design of a stabilizing tracking controller is straightforward.
	
	\item The inverse HCF transformation is used to express the control law as a state feedback of the beam variables.
\end{enumerate}

\begin{figure}
	\centering
	\def\myscale{.9}

\begin{tikzpicture}
	[scale=\myscale,
	axis/.style={thick,->,>=Stealth},
	beam/.style={thick,myBlue}]
	
	\small
	
	\def\L{4} 		
	\def\defl{0.8} 	
	\def\h{1} 		
	
	\pgfmathsetmacro{\dL}{0.25} 				
	\pgfmathsetmacro{\z}{2.8} 	
			
	\pgfmathsetmacro{\yz}{1} 		
	\pgfmathsetmacro{\lt}{0.1} 		
	
	\pgfmathsetmacro{\ka}{1.45}
	\pgfmathsetmacro{\kb}{1.08*\ka}

	\pgfmathdeclarefunction{f}{1}{\pgfmathparse{-1*\defl*(#1/\L)^2}}
	\pgfmathdeclarefunction{df}{1}{\pgfmathparse{-2*\defl*#1/(\L^2)}}
	
	\pgfmathdeclarefunction{xo}{1}{\pgfmathparse{#1 - \h/2*df(#1)/(1 + df(#1)^2)^(1/2)}}
	\pgfmathdeclarefunction{yo}{1}{\pgfmathparse{f(#1) + \h/2*1/(1 + df(#1)^2)^(1/2)}}
	
	\pgfmathdeclarefunction{xu}{1}{\pgfmathparse{#1 + \h/2*df(#1)/(1 + df(#1)^2)^(1/2)}}
	\pgfmathdeclarefunction{yu}{1}{\pgfmathparse{f(#1) - \h/2*1/(1 + df(#1)^2)^(1/2)}}

	\pgfmathsetmacro{\xLop}{xo(\L+\dL)}
	\pgfmathsetmacro{\yLop}{yo(\L+\dL)}
	
	\pgfmathsetmacro{\xLum}{xu(\L-\dL)}		
	\pgfmathsetmacro{\yLum}{yu(\L-\dL)}

	\pgfmathsetmacro{\LL}{(xo(\L+\dL) + xu(\L-\dL))/2}
	
	\pgfmathsetmacro{\xLo}{xo(\LL)}
	\pgfmathsetmacro{\yLo}{yo(\LL)}
	
	\pgfmathsetmacro{\xLu}{xu(\LL)}
	\pgfmathsetmacro{\yLu}{yu(\LL)}

	\pgfmathsetmacro{\xzop}{xo(\z+\dL)}
	\pgfmathsetmacro{\yzop}{yo(\z+\dL)}
	
	\pgfmathsetmacro{\xzum}{xu(\z-\dL)}
	\pgfmathsetmacro{\yzum}{yu(\z-\dL)}
	\pgfmathsetmacro{\xzuum}{xo(\z+\dL) - \ka*(xo(\z+\dL)-xu(\z-\dL))}
	\pgfmathsetmacro{\yzuum}{yo(\z+\dL) - \ka*(yo(\z+\dL)-yu(\z-\dL))}

	\pgfmathsetmacro{\zz}{(xo(\z+\dL) + xu(\z-\dL))/2}
	
	\pgfmathsetmacro{\xzo}{xo(\zz)}		
	\pgfmathsetmacro{\yzo}{yo(\zz)}
	
	\pgfmathsetmacro{\xzu}{xu(\zz)}
	\pgfmathsetmacro{\yzu}{yu(\zz)}
	\pgfmathsetmacro{\xzuu}{xo(\zz) - \kb*(xo(\zz)-xu(\zz))}
	\pgfmathsetmacro{\yzuu}{yo(\zz) - \kb*(yo(\zz)-yu(\zz))}

	\pgfmathsetmacro{\fLL}{f(\LL)}
	\pgfmathsetmacro{\fzz}{f(\zz)}

	\def\rC{\myscale*27}
	
	\pgfmathsetmacro{\dxt}{xo(\z+\dL)-xu(\z-\dL)}
	\pgfmathsetmacro{\dyt}{yo(\z+\dL)-yu(\z-\dL)}
	
	\pgfmathsetmacro{\dxn}{xo(\zz)-xu(\zz))}
	\pgfmathsetmacro{\dyn}{yo(\zz)-yu(\zz))}
	
	\pgfmathsetmacro{\rt}{0.45/sqrt(\dxt^2+\dyt^2)}
	\pgfmathsetmacro{\rn}{-500/sqrt(\dxn^2+\dyn^2)}
	
	\coordinate (C) at (\zz, \fzz);
	
	\coordinate (T) at ({(xo(\z+\dL)+xu(\z-\dL))/2 - \rt*\dxt},{(yo(\z+\dL)+yu(\z-\dL))/2 - \rt*\dyt}); 
	
	\coordinate (N) at ({(\xzuum-\xzop)/2 + \rn*\dxn},{(\yzuum-\yzop)/2 + \rn*\dyn});

	\pgfmathsetmacro{\alphaEnd}{atan(df(\L))}		
	\pgfmathsetmacro{\alphaMid}{atan(df(\zz))}		
	
	\def\rRA{0.2}

	\coordinate (Endpoint) at (\LL,\fLL);


	\fill[gray!40] (-0.4,-1) rectangle (0,1);
	\foreach \y in {-0.8,-0.4,0,...,0.8}
	\draw[thick] (-0.4,\y-0.2) -- (0,\y+0.2);
	\draw[thick] (0,1) -- (0,-1);
	
	\draw[axis] (-0.4,\yz) -- (\L+0.7,\yz) node[right]{$z$};
	\draw[thick] (0,\yz-\lt) -- (0,\yz+\lt) node[above]{$0$};
	\draw[thick] (\zz,\yz-\lt) -- (\zz,\yz+\lt) node[above]{$\bar z$};
	\draw[thick] (\LL,\yz-\lt) -- (\LL,\yz+\lt) node[above]{$1$};

	\draw[domain=0:\LL,samples=50,smooth,variable=\x,gray]
	plot({\x},{f(\x)});

	\draw[domain=0:\LL,samples=50,smooth,variable=\x,gray]%
	plot({xo(\x)},{yo(\x)})		
	plot({xu(\x)},{yu(\x)});	

	\draw[beam,domain=0:\L+\dL,samples=50,smooth,variable=\x]%
	plot({xo(\x)},{yo(\x)});		
	
	\draw[beam,domain=0:\L-\dL,samples=50,smooth,variable=\x]%
	plot({xu(\x)},{yu(\x)});

	\draw[gray] (\xLo,\yLo) -- (\xLu,\yLu);			
	\draw[beam] (\xLop,\yLop) -- (\xLum,\yLum);		
	
	\draw[gray] (\xzo,\yzo) -- (\xzu,\yzu);					
	\draw[gray, thick]  (\xzop,\yzop) -- (\xzum,\yzum);	    

	\draw[gray, densely dashed] (\L+1.3,0) -- (0,0);
	\draw[gray, densely dashed] (\L+1.3,\fzz) -- (\zz,\fzz);
	\draw[<-,>=Stealth,thick] (\L+1.2,\fzz) -- (\L+1.2,0) node[right, xshift=3, yshift=-5]{$w(\bar z, t)$};

	\draw[gray, densely dashed]  (\xzuu,\yzuu) -- (\xzo,\yzo);			
	\draw[gray, densely dashed]  (\xzuum,\yzuum) -- (\xzop,\yzop);			
	
	\pic[draw=black,<-,>=Stealth,thick,angle radius=\rC,angle eccentricity=1.55,"{$\varphi(\bar z, t)$}"]{angle = T--C--N};
	

	\draw[gray]	([shift={(\alphaEnd+90:\rRA)}] \LL,\fLL)
	arc[start angle=\alphaEnd+90, end angle=\alphaEnd+180, radius=\rRA];
	
	\draw[gray]	([shift={(\alphaMid-90:\rRA)}] \zz,\fzz)
	arc[start angle=\alphaMid-90, end angle=\alphaMid, radius=\rRA];

	\fill[gray]	([shift={(\alphaEnd+135:\rRA*0.618)}] \LL,\fLL) circle[radius=0.5pt];		
	\fill[gray]	([shift={(\alphaMid-45:\rRA*0.618)}] \zz,\fzz) circle[radius=0.5pt];		

	\draw[<-,>=Stealth,thick,myGreen] ($(Endpoint)+(0,0.62)$) -- ++(0,0.8)
	node[midway,right=1pt,yshift=4pt,myGreen] {$u_1(t)$};
	
	\draw[<-,>=Stealth,myRed,thick] ([shift={(-0,-0.45)}]Endpoint)
	arc[start angle=-100,end angle=10,radius=0.5];
	\node[myRed] at ([shift={(0.9,-0.5)}]Endpoint) {$u_2(t)$};

\end{tikzpicture}
	%
	%
	\caption{A Timoshenko beam clamped at $z=0$ and actuated at $z=1$.}
	\label{fig_beam}
	%
\end{figure}
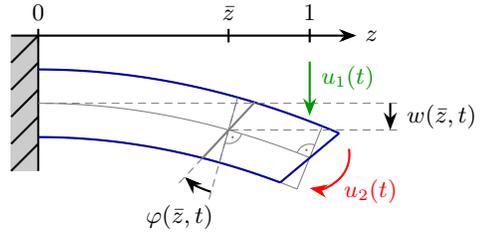

\section{Flatness-based parametrization}
\label{sec_derivFlatPar}

To derive a flatness-based parametrization for the Timoshenko beam, essentially, its solution is required.
To simplify this task, a series of transformations is applied in order to map the system into a form that is advantageous for solving the coupled PDEs \eqref{eq_waveeqs}.

\subsection{Preliminary transformations}
\label{sec_trafos}

We introduce the state $\vx(z,t)$ which allows rewriting the system \eqref{eq_waveeqs}--\eqref{eq_BCwave_1} of two second-order PDEs as a system of four first-order PDEs:
\begin{subalign}[eq_sys_x]
	\partial_t \vx^-(z,t) &= \mLambda(z) \partial_z \vx^-(z,t) \nonumber \\
	& \phantom{=} + \mA^{--}(z) \vx^-(z,t) + \mA^{-+}(z) \vx^+(z,t) \label{eq_sys_xm} \\ 
	\partial_t \vx^+(z,t) &= -\mLambda(z) \partial_z \vx^+(z,t) \nonumber \\
	& \phantom{=} + \mA^{+-}(z) \vx^-(z,t) + \mA^{++}(z) \vx^+(z,t) \label{eq_sys_xp} \\
	\vx^+(0,t) &= - \vx^-(0,t) \label{eq_BC_x_z0} \\
	\vx^-(1,t) &= \vx^+(1,t) + \mB \vu(t), \label{eq_BC_x_z1}
\end{subalign}
where $\mLambda(z) = \diagC(\lambda_1(z), \lambda_2(z))$, and $a_{ii}(z)=0$, $i=1,\dots,4$, for the diagonal elements of $\mA(z)$.
The specific expressions for $\vx(z,t)$, $\lambda_1(z)$, $\lambda_2(z)$, $\mA(z)$ and $\mB$ depend on the parameters of the Timoshenko beam, as detailed in Remark~\ref{rem_Riemann}.
The system \eqref{eq_sys_x} in Riemann coordinates comprises four heterodirectional transport equations: two propagating in the negative ($\dim \vx^-(z,t)=2$) and two in the positive $z$-direction ($\dim \vx^+(z,t)=2$).
Both subsystems \eqref{eq_sys_xm} and \eqref{eq_sys_xp} share the same absolute transport velocities $\lambda_1(z) > \lambda_2(z)>0$ as they are derived from the original pair of wave equations \eqref{eq_waveeqs}. 
{Since the transport velocities are sorted in decreasing order, the respective first transport processes associated with $x_1^\mp(z,t)$ are faster than those associated with $x_2^\mp(z,t)$.}
Defining strictly monotonically increasing positive functions $\Char_i(z) = \tint_0^z \tfrac{\dm \zeta}{\lambda_i (\zeta)}$, $i=1,2$,
together with their corresponding inverse functions $\invChar_i$ satisfying $\invChar_i(\Char_i(z))=z$ yields the system's transport times $\tau_i = \Char_i(1)$ between the boundaries $z=0$ and $z=1$, with $\tau_2 > \tau_1 > 0$.

\begin{remark}
	\label{rem_Riemann}
	The transformation%
	\begin{equation}
		\label{eq_tf_x_chi}
		\vx(z,t) = \mH^{-1}(z) \mV^{-1}(z) \begin{bmatrix}
			\partial_z w(z,t) - \varphi(z,t) \\ \partial_z \varphi(z,t) \\ \partial_t w(z,t) \\ \partial_t \varphi(z,t)
		\end{bmatrix},
	\end{equation}
	maps \eqref{eq_waveeqs}--\eqref{eq_BCwave_1} into \eqref{eq_sys_x}.
	The matrix%
	\begin{align}
		\mV(z)= \begin{bmatrix}
			\mLambda^{-1}(z) & -\mLambda^{-1}(z) \\
			\mI_2 & \mI_2
		\end{bmatrix},
	\end{align} 
	with the identity matrix $\mI_2 \in \Reals^{2 \times 2}$, is derived from the solution of an eigenvalue problem with eigenvalues $\pm\mu_1(z) = \pm\sqrt{\tfrac{\kappa(z)}{\dens(z)S(z)}}$ and $\pm\mu_2(z) = \pm\sqrt{\tfrac{E(z)}{\dens(z)}}$.
	The matrix $\mH(z) = \diagC(e^{\alpha_1(z)}, \, e^{\alpha_2(z)}, \, e^{\alpha_1(z)}, \, e^{\alpha_2(z)})$ with $\alpha_i(z) = \tint_0^z \tfrac{a_i^\prime(\zeta)}{2\lambda_i(\zeta)} \dm \zeta$, $i=1,2$, is introduced to ensure that $a_{ii}(z)=0$, $i=1,\dots,4$.
	Assuming $\mu_1(z)>\mu_2(z)$ yields transport velocities $\lambda_1(z) = \mu_1(z)$, $\lambda_2(z) = \mu_2(z)$, the coupling matrix $\mA(z) = \mH^{-1}(z) \mOmega(z) \mH(z)$ with%
	\begin{equation}
		\mOmega(z) = \tfrac 12 \begin{bmatrix}
			0 & -\lambda_1(z) & a_1(z) & -\lambda_1(z) \\
			a_3(z) & 0 & -a_3(z) & a_2(z) \\
			-a_1(z) & \lambda_1(z) & 0 & \lambda_1(z) \\
			a_3(z) & -a_2(z) & -a_3(z) & 0
		\end{bmatrix} \! ,
	\end{equation}
	where $a_1(z) = \lambda_1^\prime(z) - \tfrac{\kappa^\prime(z)}{\lambda_1(z)S(z)\dens(z)}$, $a_2(z) = \lambda_2^\prime(z) - \tfrac{\partial_z( E(z)I(z))}{\lambda_2(z)J(z)}$, $a_3(z)=\tfrac{\kappa(z)}{J(z)\lambda_1(z)}$, and the input matrix 
	$\mB = \diagC\Big( \tfrac{e^{\alpha_1(1)}}{\sqrt{\dens(1)S(1)\kappa(1)}}, \tfrac{e^{\alpha_2(1)}}{\sqrt{J(1)E(1)I(1)}} \Big)$
	in \eqref{eq_sys_x}.
	Similar expressions follow for the case $\mu_1(z) < \mu_2(z)$.
\end{remark}

%

An invertible%
\footnote{
	The kernel $\mK_I(z,\zeta) \in \Reals^{4 \times 4}$ of the inverse map ${\vx}(z,t) = \bvx(z,t) + \tint_0^z \mK_I(z,\zeta) \bvx(\zeta,t) \, \dm \zeta$ is obtained from the matrix-valued Volterra integral equation $\mK_I(z,\zeta) = \mK(z,\zeta) + \tint_\zeta^z \mK(z,\sigma) \mK_I(\sigma,\zeta) \, \dm \sigma$.
}
Volterra integral transformation
\begin{align}
	\label{eq_VIT}
	\bar{\vx}(z,t) = \vx(z,t) - \tint_0^z \mK(z,\zeta) \vx(\zeta,t) \, \dm \zeta,
\end{align}
with the kernel $\mK(z,\zeta) \in \Reals^{4 \times 4}$ defined on the triangular domain $\mathcal D = \{(z,\zeta) \in [0,1]^2 | \zeta \le z\}$, is applied, which maps the system \eqref{eq_sys_x} into the form 
\begin{subalign}[eq_sysBS]
	\partial_t \bvx^-(z,t) &= \mLambda(z) \partial_z \bvx^-(z,t) + \mA_0^-(z) \bvx^-(0,t)
	\label{eq_sysBSA}
	\\
	\partial_t \bvx^+(z,t) &= -\mLambda(z) \partial_z \bvx^+(z,t) + \mA_0^+(z) \bvx^+(0,t)
	\label{eq_sysBSB}
	\\
	\bvx^+(0,t) &= - \bvx^-(0,t) 
	\label{eq_sysBSC} \\
	\bvx^-(1,t) &= \bvu(t). \label{eq_sysBSD}
\end{subalign}
In contrast to \eqref{eq_sys_xm} and \eqref{eq_sys_xp}, the transport equations in \eqref{eq_sysBSA} and \eqref{eq_sysBSB} involve only local coupling through strictly lower triangular matrices
\begin{equation}
	\label{eq_A0}
	\mA_0^\mp(z) = \begin{bmatrix}
		0 & 0 \\
		a_0^\mp(z) & 0
	\end{bmatrix}.
\end{equation}
The new input $\bvu(t) \in \Reals^2$ is introduced via the state feedback%
\begin{align}
	\label{eq_fb_u}
	\vu(t) &= \mB^{-1} \Big( \bvu(t) - \vx^+(1,t) \\[-.5ex]
	& \phantom{=} + \tint_0^1 \mK^{--}(1,\zeta) \vx^-(\zeta,t) + \mK^{-+}(1,\zeta) \vx^+(\zeta,t) \, \dm \zeta \Big) \nonumber
\end{align}
to further simplify the system representation.

%

The kernel $\mK(z,\zeta)$ is the solution of the \emph{kernel equations}%
\begin{subalign}[eq_kernelEQS]
	\! \tmLambda(z) \partial_z \mK(z,\zeta) + \partial_\zeta ( \mK(z,\zeta) \tmLambda(z) ) &= \mK(z,\zeta) \mA(z) \label{eq_kernelEQS_PDE} \\
	(\mK^{--}(z,0) + \mK^{-+}(z,0)) \mLambda(0) &= \mA_0^-(z) \label{eq_kernelEQS_BC_0_1} \\
	-(\mK^{+-}(z,0) + \mK^{++}(z,0)) \mLambda(0) &= \mA_0^+(z) \label{eq_kernelEQS_BC_0_2} \\
	\mK(z,z) \tmLambda(z) - \tmLambda(z) \mK(z,z) &= \mA(z) \label{eq_kernelEQS_BC_z}
\end{subalign}
where $\tmLambda(z) = \diagC(\mLambda(z), -\mLambda(z))$.
Notably, \eqref{eq_kernelEQS} differs from the kernel equations commonly used for hyperbolic systems (e.g., \cite{Hu2019}), due to the modified coupling in \eqref{eq_sysBSB}, which involves the boundary value $\bvx^+(0,t)$ instead of $\bvx^-(0,t)$ as in the standard formulation.
Still, it can easily be verified that well-posedness of \eqref{eq_kernelEQS} is guaranteed and the existence of a unique piecewise $C^1(\mathcal D)$ solution $\mK(z,\zeta)$ is ensured for $\mLambda \in (C^1([0,1]))^{2 \times 2}$ and $\mA \in (C([0,1]))^{4 \times 4}$ by \cite{Hu2019}.


%


\subsection{Derivation of the parametrization}
\label{sec_flatPar}

The solution of the spatial Cauchy problem associated with the representation \eqref{eq_sysBS} is straightforward thanks to the cascaded transport equations in \eqref{eq_sysBSA} and \eqref{eq_sysBSB}.
Integration along the characteristics yields
\begin{subalign}[eq_sol]
	\bx_1^-(z,t) &= \bx_1^-(0,t + \Char_1(z)) \label{eq_sol_m1} \\
	\bx_2^-(z,t) &= \bx_2^-(0,t + \Char_2(z)) \nonumber \\
	& \phantom{=} - \tint_0^z \tfrac{a_0^-(\zeta)}{\lambda_2(\zeta)} \bx_1^-(0,t + \Char_2(z) - \Char_2(\zeta)) \, \dm \zeta \label{eq_sol_m2} \\
	\bx_1^+(z,t) &= \bx_1^+(0,t - \Char_1(z)) \label{eq_sol_p1} \\
	\bx_2^+(z,t) &= \bx_2^+(0,t - \Char_2(z)) \nonumber \\
	& \phantom{=} + \tint_0^z \tfrac{a_0^+(\zeta)}{\lambda_2(\zeta)} \bx_1^+(0,t - \Char_2(z) + \Char_2(\zeta)) \, \dm \zeta, \label{eq_sol_p2}
\end{subalign}
which depends on the boundary values $\bvx^-(0,t)$ and $\bvx^+(0,t)$, respectively.
The boundary value $\bvx^-(0,t)$ can be replaced with $\bvx^+(0,t)$ by using the BC \eqref{eq_sysBSC}, after which the solutions are rewritten in terms of $\bvx^+(1,t)$, exploiting the transport character of \eqref{eq_sysBSA} and \eqref{eq_sysBSB}.
This allows us to parametrize the state $\bvx(\cdot,t)$ by the flat output 
\begin{equation}
	\label{eq_def_y}
	\vy(t) = \bvx^+(1,t)
\end{equation}
resulting in
\begin{subalign}[eq_flatPar_x]
	\bx_1^-(z,t) &= - y_1(t + \tau_1 + \Char_1(z)) \label{eq_flatPar_xm1} \\
	\bx_2^-(z,t) &= - y_2(t + \tau_2 + \Char_2(z)) \nonumber \\
	& \;\, + \tint_0^1  \tfrac{a_0^+(\zeta)}{\lambda_2(\zeta)} y_1(t + \tau_1 + \Char_2(z) + \Char_2(\zeta)) \, \dm \zeta \nonumber \\
	& \;\, + \tint_0^z  \tfrac{a_0^-(\zeta)}{\lambda_2(\zeta)} y_1(t + \tau_1 + \Char_2(z) - \Char_2(\zeta)) \, \dm \zeta \label{eq_flatPar_xm2}
	\\
	\bx_1^+(z,t) &= y_1(t + \tau_1 - \Char_1(z)) \label{eq_flatPar_xp1} \\
	\bx_2^+(z,t) &= y_2(t + \tau_2 - \Char_2(z)) \nonumber \\
	& \;\, - \tint_z^1 \tfrac{a_0^+(\zeta)}{\lambda_2(\zeta)} y_1(t + \tau_1 - \Char_2(z) + \Char_2(\zeta)) \, \dm \zeta. \label{eq_flatPar_xp2}
\end{subalign}
The flatness-based parametrization of the input
\begin{subalign}[eq_flatPar_u]
	\bu_1(t) &= - y_1(t + 2 \tau_1) \label{eq_flatPar_u1} \\
	\bu_2(t) &= - y_2(t + 2 \tau_2) \nonumber \\
	& \phantom{=} + \tint_0^{\tau_2}  \ba_0^+(\tau) \, y_1(t + \tau_1 + \tau_2 + \tau) \, \dm \tau \nonumber \\
	& \phantom{=} + \tint_0^{\tau_2}  \ba_0^-(\tau) \, y_1(t + \tau_1 + \tau_2 - \tau) \, \dm \tau, \label{eq_flatPar_u2}
\end{subalign}
where $\ba_0^\mp(\tau) = a_0^\mp(\invChar_2(\tau))$, is obtained by substituting \eqref{eq_flatPar_xm1} and \eqref{eq_flatPar_xm2} into the BC \eqref{eq_sysBSD}.

The flatness-based parametrization \eqref{eq_flatPar_x} and \eqref{eq_flatPar_u} of the state and input of the hyperbolic system involves the flat output $\vy(t)$ at time $t$ as well as predictions $\vy(t+\tau)$, $\tau>0$, without any delays $\vy(t-\tau)$, $\tau>0$.
Specifically, it involves $y_1(t+\tau)$ for $\tau \in [0, \tau_1 + 2\tau_2]$, and $y_2(t+\tau)$ for $\tau \in [0, 2\tau_2]$.

\begin{remark}
	\label{rem_y_x}
	An intuitive interpretation of the flat output in terms of the original beam variables $w(z,t)$ and $\varphi(z,t)$ is not obvious, since according to \eqref{eq_VIT} and \eqref{eq_def_y}, $\vy(t) = \vx^+(1,t) - \tint_0^1 \mK^{+-}(1,\zeta) \vx^-(\zeta,t) + \mK^{++}(1,\zeta) \vx^+(\zeta,t) \, \dm \zeta$, where $\vx(z,t)$ itself already represents an exponentially scaled linear combination of $w(z,t)$, $\varphi(z,t)$ and their partial derivatives as defined in \eqref{eq_tf_x_chi}.
\end{remark}


\section{Hyperbolic Controller Form (HCF)}
\label{sec_HCF}

To obtain the HCF for the Timoshenko beam based on the input parametrization \eqref{eq_flatPar_u}, we define the HCF state
\begin{subalign}[eq_def_HCF_state]
	\eta_1(\tau,t) &= y_1(t + \tau), \quad \tau \in [0, 2\tau_1] \\
	\eta_2(\tau,t) &= y_2(t + \tau), \quad \tau \in [0, 2\tau_2]
\end{subalign}
in view of \eqref{eq_flatPar_x} and \eqref{eq_flatPar_u}.
Then, the components $\eta_i(\tau,t)$, $i=1,2$, satisfy homogeneous
transport equations $\partial_t \eta_i(\tau,t) = \partial_\tau \eta_i(\tau,t)$ with normalized, negative transport velocities on domains $\tau \in [0, 2\tau_i)$.
The respective outputs constitute the flat output components, i.e., $\eta_i(0,t) = y_i(t)$.
The BCs of the HCF are obtained from the input parametrization \eqref{eq_flatPar_u}, where expressions $y_i(t+\tau)$, $i=1,2$, for $\tau \in [0, 2\tau_i]$ can be replaced by the HCF state in \eqref{eq_def_HCF_state}.
However, for values $y_1(t+\tau)$ with prediction amplitudes $\tau > 2\tau_1$, this substitution is no longer valid.
Instead, applying a forward time shift to \eqref{eq_flatPar_u1} reveals that these values can be replaced by predictions of $u_1(t)$.
Hence, using the substitution 
\begin{subalign}[eq_y_to_eta]
	y_1(t+\tau) &= \begin{cases}
		\eta_1(\tau,t), & \tau \in [0, 2 \tau_1] \label{eq_y_to_eta1} \\
		- \bu_1(t + \tau - 2\tau_1), & \tau > 2\tau_1
	\end{cases} \\
	y_2(t + \tau) &= \eta_2(\tau,t), \qquad \qquad \qquad \tau \in [0, 2 \tau_2] \label{eq_y_to_eta2}
\end{subalign}
in \eqref{eq_flatPar_u} and subsequently solving for the boundary values $\eta_i(2\tau_i,t)$ provides the BCs of the HCF
\begin{subalign}[eq_HCF]
	\partial_t \eta_1(\tau,t) &= \partial_\tau \eta_1(\tau,t), \quad \quad \tau \in [0, 2\tau_1) \label{eq_HCF_PDE_1} \\
	\partial_t \eta_2(\tau,t) &= \partial_\tau \eta_2(\tau,t), \quad \quad \tau \in [0, 2\tau_2) \label{eq_HCF_PDE_2} \\
	\eta_1(2\tau_1, t) &= - \bu_1(t) \label{eq_HCF_BC_1} \\
	\eta_2(2\tau_2, t) &= - \bu_2(t) + \tint_{\tau_1}^{2\tau_1} \ba_0^-(\tau_1 + \tau_2 - \tau) \eta_1(\tau,t) \, \dm \tau \nonumber \\
	& \phantom{=} - \tint_{0}^{\Delta \tau + \tau_2} \ta_0(\tau) \bu_1(t + \tau)  \, \dm \tau, \label{eq_HCF_BC_2}
\end{subalign}
where $\Delta \tau = \tau_2 - \tau_1$ and $\ta_0(\tau) = \ba_0^-(\Delta \tau - \tau) h(\Delta \tau - \tau) + \ba_0^+(\tau - \Delta \tau) h (\tau - \Delta \tau)$ with the Heaviside function $h$.
Since the HCF consists of two transport equations where the input acts only through the BCs, the control design in this state representation is particularly simple.
\begin{Lemma}[HCF]
	\label{lem_HCF}
	There exists an invertible transformation that maps the state $\vx(\cdot,t)$ of \eqref{eq_sys_x} into the state $\veta(\cdot,t)$ of the HCF in \eqref{eq_HCF}.
\end{Lemma}
\begin{proof}
	The mappings between the HCF state $\veta(\cdot,t)$ and the state $\bvx(\cdot,t)$ are provided in the Appendix.
	The states $\bvx(\cdot,t)$ and $\vx(\cdot,t)$ are related through the invertible Volterra integral transformation \eqref{eq_VIT}.
\end{proof}

The HCF \eqref{eq_HCF} for the Timoshenko beam does not constitute a state representation in the classical sense, as it involves positive time shifts of the input.
Instead, we refer to \eqref{eq_HCF} as a \emph{generalized HCF with input predictions} (cf. \cite{Ecklebe2025}).
Since the system dynamics are equivalently represented by the HCF and the original state representation \eqref{eq_sys_x}, the corresponding mapping between the associated states also involves the input (see Appendix).



\section{State feedback tracking controller}
\label{sec_ctrl}

To design a tracking controller that stabilizes the transition of the Timoshenko beam, we assume given reference trajectories 
\begin{equation}
	\label{eq_yref}
	y_{i,\mathrm{r}}(t) = \begin{cases}
		y_{i,0}, & t < t_0 \\
		p_i(t), & t_0 \le t \le t_T\\ 
		y_{i,T}, & t > t_T,
	\end{cases}
\end{equation}
$i=1,2$, for the flat output $\vy(t)$, where $p_i(t)$ are continuous functions satisfying $p_i(t_0)=y_{i,0}$ and $p_i(t_T)=y_{i,T}$. 
The trajectory represents the transition of the system between the stationary solutions $(w_0(z), \varphi_0(z))$ and $(w_T(z), \varphi_T(z))$ of \eqref{eq_waveeqs}.
The associated initial and terminal values $\vy_0$ and $\vy_T$ are obtained via the preliminary transformations from Section~\ref{sec_derivFlatPar} (cf. Remark~\ref{rem_y_x}).
Notably, by evaluating the input parametrization \eqref{eq_flatPar_u} for $\vy(\cdot)=\vy_\mathrm{r}(\cdot)$, the corresponding feedforward control $\vu_\mathrm{r}(t)$ can be derived.
To stabilize the reference trajectory \eqref{eq_yref}, a tracking controller is derived. 
For this purpose, a new input $\vv(t)$ is introduced via the state feedback
\begin{subalign}[eq_decouplingFB]
	\bu_1(t) &= - v_1(t) \\
	\bu_2(t) &= - v_2(t) + \tint_{\tau_1}^{2\tau_1} \ba_0^-(\tau_1 + \tau_2 - \tau) \eta_1(\tau,t) \, \dm \tau \nonumber \\
	& \phantom{=} + \tint_{0}^{\Delta \tau + \tau_2} \ta_0(\tau) v_1(t + \tau)  \, \dm \tau,
\end{subalign}
which decouples the HCF \eqref{eq_HCF} into two independent systems of transport equations
\begin{subalign}[eq_Bruno]
	\partial_t \eta_i(\tau,t) &= \partial_\tau \eta_i(\tau,t), \quad \quad \tau \in [0, 2\tau_i) \\
	\eta_i(2\tau_i,t) &= v_i(t) \label{eq_Bruno_BC}
\end{subalign}
for $i=1,2$.
This structure can be interpreted as the hyperbolic counter-part to the \emph{Brunovsky} form (see, e.g., \cite{Isidori1985}).
The control design for the system in the form \eqref{eq_Bruno} is particularly simple, since it essentially corresponds to designing controllers for two independent SI systems (see, e.g., \cite{Woittennek2013}).
To this end, we define the tracking errors
\begin{equation}
	\label{eq_error}
	e_i(t + \tau) = \eta_i(\tau,t) - y_{i,\mathrm{r}}(t + \tau), \quad i=1,2.
\end{equation}
For simplicity, consider decoupled error dynamics
\begin{equation}
	\label{eq_errorDyn}
	e_i(t+2\tau_i) + \gamma_i e_i(t) + \tint_{0}^{2\tau_i} \alpha_i(\tau) e_i(t+\tau) \, \dm \tau = 0
\end{equation}
for the closed-loop system, which involve both lumped and distributed evaluations of the error components.
It is shown in \cite[Theorem~4.1]{Russell1991} that each of the SI subsystems $i=1,2$ in \eqref{eq_errorDyn} is a Riesz spectral system if $\gamma_i \ne 0$ and $\alpha_i(\tau)$ is piecewise continuous on the corresponding interval $[0, 2\tau_i]$.
If these conditions hold for both subsystems, exponential stability of the closed-loop dynamics \eqref{eq_errorDyn} and thus convergence to the reference is guaranteed if the roots $s_k \in \mathbb C$, $k\in \mathbb N$, of $\mathrm{e}^{2\tau_i s} + \gamma_i + \tint_{0}^{2 \tau_i} \alpha_i(\tau) \mathrm{e}^{ \tau s} \, \dm \tau $
satisfy $\Re(s_k)<0$ (see \cite[Theorem~3.2.8]{Curtain2020}).
As choosing appropriate parameters $\gamma_i$ as well as functions $\alpha_i(\tau)$, $i=1,2$, is not the focus of this contribution, without loss of generality, we henceforth assume $\alpha_i(\tau) = 0$.
In this case, $|\gamma_i| < 1$ guarantees exponential stability and the subsequent expressions are more concise and readable.

For $\alpha_i(\tau) = 0$ and $|\gamma_i| < 1$, $i=1,2$, applying
\begin{equation}
	\label{eq_FB_v}
	v_i(t) = y_{i,\mathrm{r}}(t+2\tau_i) - \gamma_i(\eta_i(0,t) - y_{i,\mathrm{r}}(t))
\end{equation}
to \eqref{eq_Bruno} results in the decoupled error dynamics \eqref{eq_errorDyn}.
To obtain the corresponding control law for the input $\bvu(t)$, the feedback \eqref{eq_FB_v} is inserted into the input transformation \eqref{eq_decouplingFB}.
This substitution, however, also requires
\begin{align}
	\label{eq_v_pred}
	v_1(t+\tau) & = y_{1,\mathrm{r}}(t+2\tau_1+\tau) - \gamma_1 \big( \eta_1(0,t+\tau) \nonumber \\
	& \phantom{=} - y_{1,\mathrm{r}}(t+\tau) \big)
\end{align}
for $\tau \in [0,\Delta \tau + \tau_2]$.
Predictions $y_{1,\mathrm{r}}(t+\sigma)$, $\sigma>0$, of the reference trajectory are readily available, regardless of the prediction amplitude (cf. \eqref{eq_yref}).
The term $\eta_1(0,t+\tau)$ is replaced using
\begin{equation}
	\label{eq_subs_eta1_pred}
	\eta_1(0,t+\tau) = \left\{ \begin{aligned}
		& \eta_1(\tau,t), && \tau \in [0,2\tau_1] \\
		-& \gamma_1 \eta_1(0,t+\tau-2\tau_1) \\
		&+ f(y_{1,\mathrm{r}}(\cdot)),
		&& \tau > 2\tau_1, \\
	\end{aligned} \right.
\end{equation}
where $f(y_{1,\mathrm{r}}(\cdot))$ collects all reference-related terms.
For $\tau \in [0,2\tau_1]$, the prediction $\eta_1(0,t+\tau)$ is directly expressed in terms of the state component $\eta_1(\tau,t)$ by exploiting the transport property of \eqref{eq_HCF_PDE_1}, which also holds in the closed loop.
If $\tau > 2\tau_1$, the error dynamics \eqref{eq_errorDyn} for the first error component together with the definition \eqref{eq_error} are employed.
This reduces the prediction amplitude by $2\tau_1$.
Importantly, the substitution rule \eqref{eq_subs_eta1_pred} can be applied successively.
It terminates after a finite number of steps.
Hence, $\eta_1(0,t+\tau)$ can always be expressed in terms of the state component $\eta_1(\tau,t)$ and the reference $y_{i,\mathrm{r}}(t)$.
%
%
%
%
Therefore, the feedback that results from inserting \eqref{eq_FB_v} and \eqref{eq_v_pred} with \eqref{eq_subs_eta1_pred} into \eqref{eq_decouplingFB}, is a feedback of the HCF state $\veta(\cdot,t)$.

\begin{remark}
	\label{rem_quasistat}
	As the HCF \eqref{eq_HCF} involves input predictions $\bu_1(t+\tau)$, $\tau>0$, and thus is not a state representation in the classical sense, the decoupling feedback \eqref{eq_decouplingFB} likewise depends on predicted input values and is therefore non-static.
	It can be interpreted as \emph{quasi-static} feedback in the context of hyperbolic systems.
\end{remark}

Finally, to express the control law in the original beam coordinates, the transformations in Section~\ref{sec_trafos} and the Appendix leading to the HCF are applied in reverse order:
First, the state transformation \eqref{eq_TF_eta_bx} is used to substitute $\veta(\cdot,t)$ with $\bvx(\cdot,t)$, which can in turn be expressed by $\vx(\cdot,t)$ using \eqref{eq_VIT}.
Next, the feedback for $\bvu(t)$ is inserted into \eqref{eq_fb_u} to obtain the control law for $\vu(t)$, and finally, \eqref{eq_tf_x_chi} from Section~\ref{sec_trafos} is applied to replace $\vx(\cdot,t)$ with the beam variables $w(z,t)$, $\varphi(z,t)$ and their partial derivatives.
The main result is summarized in the following theorem.
\begin{theorem}
	\label{th_ctrl}	
	The controller comprising \eqref{eq_decouplingFB} and \eqref{eq_FB_v}--\eqref{eq_subs_eta1_pred}, with parameters $|\gamma_i|<1$, $i=1,2$, ensures exponential convergence of the state $\veta(\cdot,t)$ of the HCF \eqref{eq_HCF} to the reference $\vy_{\mathrm{r}}(\cdot)$, and hence of $\vx(\cdot,t)$ in \eqref{eq_sys_x} to $\vx_{\mathrm{r}}(\cdot,t)$.
\end{theorem}
\begin{proof}
	Applying \eqref{eq_decouplingFB} and \eqref{eq_FB_v}--\eqref{eq_subs_eta1_pred} to the HCF \eqref{eq_HCF} results in the closed-loop error dynamics \eqref{eq_errorDyn} with $\alpha_i(\tau)=0$, $i=1,2$.
	Choosing $|\gamma_i|<1$ ensures that $e_i(t)$ converges to zero and thus $\eta_i(\tau,t)$ to $\eta_{i,\mathrm{r}}(\tau,t)=y_{i,\mathrm{r}}(t+\tau)$, $i=1,2$, for $t \to \infty$.
	This implies exponential convergence of $\vx(z,t)$ to $\vx_{\mathrm{r}}(z,t)$ according to Lemma~\ref{lem_HCF}, where $\vx_{\mathrm{r}}(z,t)$ is obtained by inserting $\vy_{\mathrm{r}}(t)$ into \eqref{eq_flatPar_x} and the inverse of \eqref{eq_VIT}. 
\end{proof}


\section{Simulation results}
\label{sec_sim}

We validate the proposed control design for the Timoshenko beam \eqref{eq_waveeqs}--\eqref{eq_BCwave_1} with constant parameters $\kappa=1.25$, $S=1$, $\rho = 0.8$, $EI = 0.5$ and $J=0.98$ in simulation.
These parameters yield propagation speeds $\lambda_1 = \tfrac 54$ and $\lambda_2 = \tfrac 57$ in \eqref{eq_sys_x}, corresponding to transport times $\tau_1 = 0.8$ and $\tau_2 = 1.4$.
The reference trajectory is defined as in \eqref{eq_yref} with third-degree polynomials $p_i(z)$, $i=1,2$, as well as $t_0 = 4.5$, $T=t_T-t_0=1$, $\vy_0 = \vect 0$ and $\vy_T=[0.1,\,-0.2]^T$.
Trajectory tracking is achieved using the feedback law from Theorem~\ref{th_ctrl}.
For simplicity, the control parameters $\gamma_1 = \gamma_2 = 0$ are chosen, ensuring finite-time closed-loop stability.
Even in that case, due to \eqref{eq_decouplingFB} and \eqref{eq_fb_u}, the feedback involves the entire system state.

The Timoshenko beam is simulated with ICs $x_{i,0}^\mp(z)=\pm 0.2\sin^2(2\pi i z)$, $i=1,2$.
Figure~\ref{fig_res_y} illustrates the results for the flat output $\vy(t)$.
For each component $y_i(t)$, $i=1,2$, the initial errors due to $\vx(z,0)\ne \vect 0$ are compensated within the time interval $[0,2\tau_i)$, demonstrating finite-time convergence as $e_i(t)=0$ for $t \ge 2\tau_i$.
Subsequently, the controller stabilizes the transition of the system.
For demonstration purposes the starting time $t_0$ is chosen sufficiently large so that $\ve(t)=\vect 0$ before the transition starts.
For $t\ge t_T=5.5$, the transition is completed, and the flat output remains constant at $\vy(t)=\vy_T$.
Figure~\ref{fig_res_x} shows the corresponding results for the transverse beam displacement $w(z,t)$, $z\in[0,1]$, where the terminal stationary configuration $w_T(z)$ corresponding to $\vy_T$ can be clearly observed for $t\ge 5.5$.
The evolution of the rotation angle $\varphi(z,t)$ is omitted.


\begin{figure}[t]
	\centering
	\includegraphics[width=\columnwidth]{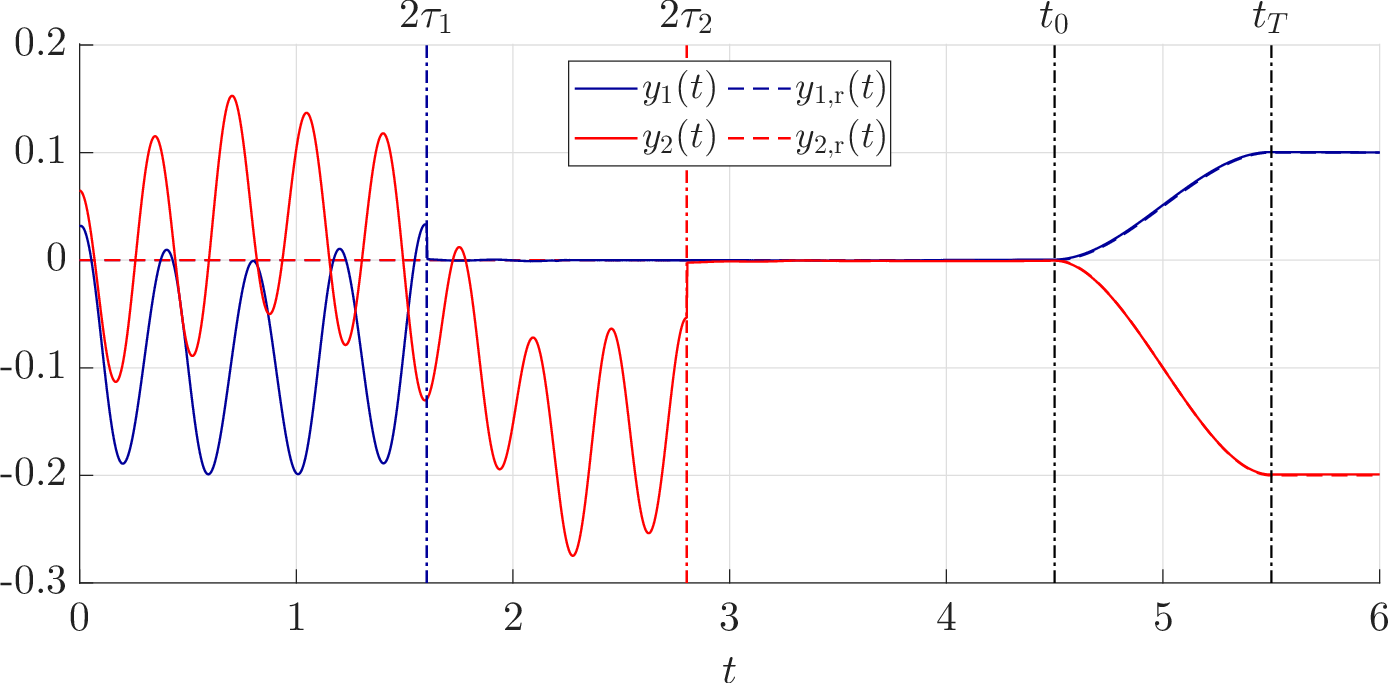}	%
	\caption{Convergence of the flat output $\vy(t)$ to its reference $\vy_{\mathrm{r}}(t)$.}
	\label{fig_res_y}
\end{figure}

\begin{figure}[t]
\centering
\includegraphics[width=\columnwidth]{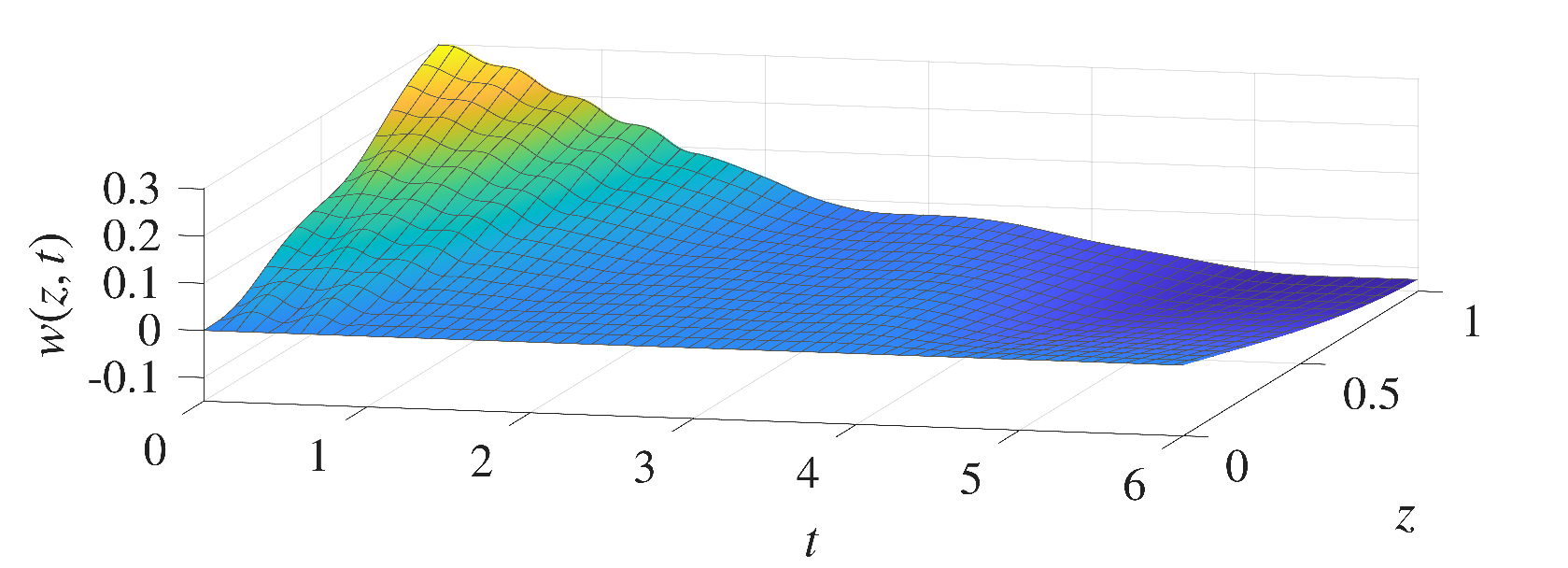} %
\caption{Beam displacement $w(z,t)$ during the transition between two stationary solutions.}
\label{fig_res_x}
\end{figure}

\begin{figure*}
	\begin{subalign}[eq_TF_bx_eta]		
		\bx_1^+(z,t) &= \eta_1(\tau_1 - \Char_1(z), t) \label{eq_TF_bx_eta_p1} \\
		\bx_2^+(z,t) &= \eta_2(\tau_2 - \Char_2(z), t) - \tint_{\tau_1}^{\min(2\tau_1, \tau_1 + \tau_2 - \Char_2(z))} \ba_{0,1}^+ (z,\tau) \eta_1(\tau, t) \, \dm \tau + \tint_{\min(0,\Delta \tau - \Char_2(z))}^{\Delta \tau - \Char_2(z)} \ba_{0,2}^+ (z, \tau ) \bu_1(t + \tau) \, \dm \tau \label{eq_TF_bx_eta_p2} \\
		\bx_1^-(z,t) &= - \eta_1(\tau_1 + \Char_1(z), t) \label{eq_TF_bx_eta_m1} \\
		\bx_2^-(z,t) &= - \eta_2(\tau_2 + \Char_2(z), t) + \tint_{\tau_1 + \Char_2(z)}^{\max(2\tau_1, \tau_1 + \Char_2(z))}  \ba_{0,3}^+ (z,\tau ) \eta_1(\tau, t) \, \dm \tau + \tint_{\tau_1}^{\min(2\tau_1, \tau_1 + \Char_2(z))}  \ba_{0,1}^-(z,\tau) \eta_1(\tau, t) \, \dm \tau \nonumber \\
		& \phantom{=} - \tint_{\max(0,\Char_2(z) - \tau_1)}^{\Delta \tau + \Char_2(z)}  \ba_{0,4}^+ (z,\tau ) \bu_1(t + \tau) \, \dm \tau - \tint_{\min(0,\Char_2(z)-\tau_1)}^{\Char_2(z)-\tau_1}  \ba_{0,2}^-(z, \tau) \bu_1(t + \tau) \, \dm \tau \label{eq_TF_bx_eta_m2}
	\end{subalign}
	\vspace*{-10pt}
	\begin{subalign}[eq_TF_eta_bx]		
		\eta_1(\tau,t) &= \begin{cases}
			\bx_1^+(\invChar_1(\tau_1-\tau),t), & \tau \in [0,\tau_1) \\
			- \bx_1^-(\invChar_1(\tau - \tau_1),t), & \tau \in [\tau_1, 2\tau_1)
		\end{cases} \label{eq_TF_eta_bx_1} \\
		\eta_2(\tau,t) &= \left\{ \begin{aligned}
			& \bx_2^+(\invChar_2(\tau_2-\tau),t) - \tint_{\tau_1}^{\min(2\tau_1,\tau_1+\tau)} b_{0,1}^+(\tau,\sigma) \bx_1^-(\invChar_1(\sigma - \tau_1),t) \, \dm \sigma \\
			& \qquad - \tint_{\min(0,\tau-\tau_1)}^{\tau-\tau_1} b_{0,2}^+(\tau,\sigma) \bu_1(t+\sigma) \, \dm \sigma,
			&& \quad \tau \in [0,\tau_2) \\
			& - \bx_2^-(\invChar_2(\tau - \tau_2),t) - \tint_{\tau - \Delta \tau}^{\max(2\tau_1,\tau - \Delta \tau)} b_{0,1}^+(\tau,\sigma) \bx_1^-(\invChar_1(\sigma - \tau_1),t) \, \dm \sigma \\
			& \qquad - \tint_{\tau_1}^{\min(2\tau_1,\tau - \Delta \tau)} b_{0,1}^-(\tau,\sigma) \bx_1^-(\invChar_1(\sigma - \tau_1),t) \, \dm \sigma \\
			& \qquad - \tint_{\max(0,\tau - \tau_2 - \tau_1)}^{\tau-\tau_1} b_{0,2}^+(\tau,\sigma) \bu_1(t+\sigma) \, \dm \sigma  - \tint_{\min(0,\tau - \tau_2 - \tau_1)}^{\tau - \tau_2 - \tau_1} b_{0,2}^-(\tau,\sigma) \bu_1(t+\sigma) \, \dm \sigma,
			&& \quad \tau \in [\tau_2, 2\tau_2)
		\end{aligned} \right. \label{eq_TF_eta_bx_2}
	\end{subalign}
	\vspace*{-15pt}
\end{figure*}


\section{Concluding remarks}
\label{sec_concl}

This paper presents a flatness-based control design approach for hyperbolic systems that employs the HCF, a special state representation that greatly simplifies the design of stabilizing feedback laws.
The concept is demonstrated for the classical Timoshenko beam, but can be extended to general hyperbolic systems.
Since the HCF derived here constitutes a generalized state representation, a particular point of future research is whether a classical state representation free of input predictions can be obtained through an alternative choice of the HCF state.



\section*{Appendix: HCF state transformation}

The mapping from $\veta(\cdot,t)$ to $\bvx(z,t)$ is obtained by applying \eqref{eq_y_to_eta} to \eqref{eq_flatPar_x}, which yields \eqref{eq_TF_bx_eta} with kernel functions
$\ba_{0,1}^+ (z,\tau) = \ba_0^+ (\Char_2(z) - \tau_1 + \tau )$, 
$\ba_{0,2}^+ (z,\tau) = \ba_0^+ (\Char_2(z) + \tau_1 + \tau )$, 
$\ba_{0,3}^+ (z,\tau) = \ba_0^+ (-\Char_2(z) - \tau_1 + \tau )$, 
$\ba_{0,4}^+ (z,\tau) = \ba_0^+ (-\Char_2(z) + \tau_1 + \tau )$, 
$\ba_{0,1}^- (z,\tau) = \ba_0^- (\Char_2(z) + \tau_1 - \tau )$, and 
$\ba_{0,2}^- (z,\tau) = \ba_0^- (\Char_2(z) - \tau_1 - \tau )$.

Based on this result, the inverse mapping \eqref{eq_TF_eta_bx} can be derived, where 
$b_{0,1}^+(\tau,\sigma) = \ba_0^+(\Delta \tau - \tau + \sigma)$, 
$b_{0,2}^+(\tau,\sigma) = \ba_0^+(\tau_1 + \tau_2 - \tau + \sigma)$, 
$b_{0,1}^-(\tau,\sigma) = \ba_0^-(-\Delta \tau + \tau - \sigma)$, and 
$b_{0,2}^-(\tau,\sigma) = \ba_0^-(-\tau_1 - \tau_2 + \tau - \sigma)$.
The mapping \eqref{eq_TF_eta_bx_1} from $\bvx(z,t)$ to the first HCF state component $\eta_1(\cdot,t)$ follows from evaluating \eqref{eq_TF_bx_eta_p1} for $z = \invChar_1(\tau_1 - \tau)$ and \eqref{eq_TF_bx_eta_m1} for $z = \invChar_1(\tau - \tau_1)$, respectively.
The mapping \eqref{eq_TF_eta_bx_2} for the second HCF state component is derived from \eqref{eq_TF_bx_eta_p2} and \eqref{eq_TF_bx_eta_m2} by replacing $\eta_1(\cdot,t)$ with $\bx_1^\mp(\cdot,t)$ using \eqref{eq_TF_eta_bx_1} followed by analogous substitutions for $z$ as for \eqref{eq_TF_bx_eta_p1} and \eqref{eq_TF_bx_eta_m1}.

\printbibliography

\end{document}